\documentclass[a4paper,11pt,reqno]{amsart}

\usepackage[margin=1in,includehead,includefoot,asymmetric,marginparwidth=1cm]{geometry}
\usepackage[utf8]{inputenc}
\usepackage[T1]{fontenc}
\usepackage{lmodern,microtype}
\usepackage{hyperref}
\usepackage{url}
\usepackage{booktabs}
\usepackage{mathtools,amssymb,amsthm}
\usepackage[foot]{amsaddr}
\usepackage{graphicx, color}
\usepackage[margin=5mm]{caption}
\usepackage{enumitem}
\usepackage[b]{esvect} 
\usepackage[textsize=footnotesize,disable]{todonotes}
\usepackage{mdframed}
\usepackage{xpatch}
\usepackage{subcaption}
\usepackage{mycommands}
\graphicspath{{./graphics/}}

\newtheorem{theorem}{Theorem}

\newtheorem{lemma}[theorem]{Lemma}
\newtheorem{corollary}[theorem]{Corollary}
\newtheorem{problem}[theorem]{Problem}
\theoremstyle{remark}
\newtheorem{remark}[theorem]{Remark}


\hypersetup{
  pdftitle={Listing spanning trees of outerplanar graphs by pivot-exchanges},
  pdfauthor={Nastaran Behrooznia, Torsten M\"utze}
}

\title{Listing spanning trees of outerplanar graphs by~pivot-exchanges}

\author{Nastaran Behrooznia}
\address[Nastaran Behrooznia]{Department of Computer Science, University of Warwick, United Kingdom}
\email{nastaran.behrooznia@warwick.ac.uk}

\author{Torsten M\"utze}
\address[Torsten M\"utze]{Institut f\"ur Mathematik, Universit\"at Kassel, Germany \& Department of Theoretical Computer Science and Mathematical Logic, Charles University, Prague, Czech Republic}
\email{tmuetze@mathematik.uni-kassel.de}

\thanks{This work was supported by Czech Science Foundation grant GA~22-15272S. Both authors participated in the workshop `Combinatorics, Algorithms and Geometry' in March 2024, which was funded by German Science Foundation grant~522790373.}

\begin{document}

\begin{abstract}
We prove that the spanning trees of any outerplanar triangulation~$G$ can be listed so that any two consecutive spanning trees differ in an exchange of two edges that share an end vertex.
For outerplanar graphs~$G$ with faces of arbitrary lengths (not necessarily~3) we establish a similar result, with the condition that the two exchanged edges share an end vertex or lie on a common face.
These listings of spanning trees are obtained from a simple greedy algorithm that can be implemented efficiently, i.e., in time~$\cO(n \log n)$ per generated spanning tree, where $n$ is the number of vertices of~$G$.
Furthermore, the listings correspond to Hamilton paths on the 0/1-polytope that is obtained as the convex hull of the characteristic vectors of all spanning trees of~$G$.
\end{abstract}

\maketitle

\section{Introduction}
\label{sec:intro}

For a given graph~$G$, let~$\cT(G)$ denote the set of all spanning trees of~$G$.
Two spanning trees of~$G$ differ in an \defi{edge exchange} if the symmetric difference of their edge sets is a 2-element set, i.e., each of the two spanning trees is obtained from the other one by removing one edge and adding another.
The \defi{flip graph} $\cF(G)$ has the set~$\cT(G)$ as vertex set, and an edge between any two spanning trees that differ in an edge exchange; see Figure~\ref{fig:span}.
It is well-known that~$\cF(G)$ is the skeleton of the 0/1-polytope that is obtained as the convex hull of all characteristic vectors~$\chi(T)$ of spanning trees~$T\in\cT(G)$ (see~\cite[Thm.~40.6]{MR1956925}).
Specifically, if the edge set of~$G$ is~$\{1,\ldots,m\}$, then for all $e\in\{1,\ldots,m\}$ the characteristic vector~$\chi(T)$ has a 1-bit at position~$e$ if the edge~$e$ belongs to~$T$, and a 0-bit at position~$e$ otherwise.

\begin{figure}[h!]
\includegraphics{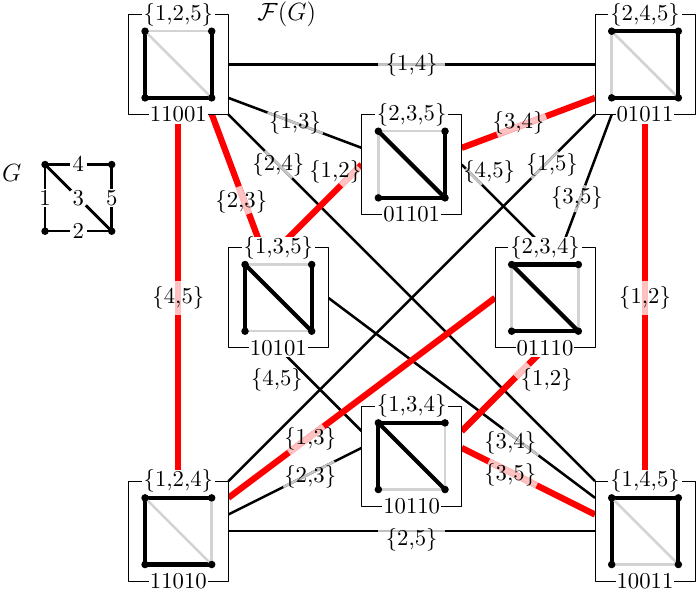}
\caption{The spanning tree flip graph~$\cF(G)$ for the `diamond' graph~$G$ on the left, with a Hamilton cycle highlighted.
For each spanning tree, the set of edges is shown above, and the characteristic vector is shown below.
Edges of~$\cF(G)$ are labelled by the two edges of~$G$ being exchanged.}
\label{fig:span}
\end{figure}

We are interested in computing Hamilton paths or cycles in the flip graph~$\cF(G)$, i.e., we aim to list all spanning trees of~$G$ such that any two consecutive spanning trees differ in an edge exchange.
Such a listing of combinatorial objects subject to some small-change condition is generally referred to as a \defi{Gray code}~\cite{MR1491049,MR4649606}.
A Gray code is called \defi{cyclic} if the first and last object also differ in the small-change condition, i.e., this corresponds to a Hamilton cycle in the flip graph.

Cummins~\cite{MR245357} first proved that~$\cF(G)$ admits a Hamilton cycle for any graph~$G$.
He showed more generally that for any prescribed edge of~$\cF(G)$, there is a Hamilton cycle containing this edge.
Similar results were obtained by Shank~\cite{shank_1968}, Kamae~\cite{MR0246789}, and Kishi and Kajitani~\cite{kishi_kajitani_1968}.

Harary and Holzmann~\cite{MR0307952} proved more generally that the base exchange flip graph of any matroid has a Hamilton cycle.
They showed this in the stronger sense that any edge of this flip graph can be prescribed to be contained in the cycle, and to be avoided by another cycle.
Even more generally, Naddef and Pulleyblank~\cite{MR638286,MR762893} proved that the skeleton of any 0/1-polytope either admits a Hamilton path between any two prescribed end vertices, or it is a hypercube, in which case it admits a Hamilton path between any two end vertices of opposite parity.

Algorithmically, a Hamilton path in~$\cF(G)$ can be computed in time~$\cO(1)$ on average per generated tree using Smith's algorithm~\cite{smith_1997} (this is Algorithm~S in~\cite[Sec.~7.2.1.6]{MR3444818}); see Figure~\ref{fig:3GCs}~(a).

Given these strong Hamiltonicity properties of the graph~$\cF(G)$, there has recently been interest to strengthen them by restricting the allowed edge exchanges.

\subsection{The pivot-exchange property}
\label{sec:vp}

We introduce some more notation.
For a graph~$G$, we write $V(G)$ and~$E(G)$ for set of vertices and edges of~$G$, respectively.
For a subgraph $H\seq G$ and edges $e\in E(H)$ and $f\in E(G)\setminus E(H)$, we write $H-e$ and $H+f$ for the graphs obtained from~$H$ by removing and adding the edges~$e$ and~$f$, respectively.
We will think of a subgraph~$H$ as a subset of edges from~$E(G)$, so the operations~$H-e$ and $H+f$ remove and add an element from the set, respectively.
Formally, an \defi{edge exchange} for a spanning tree~$T\in\cT(G)$ is a pair~$\{e,f\}$ of edges such that~$e\in E(T)$ and $f\in E(G)\setminus E(T)$ with $T-e+f=T\triangle \{e,f\}\in \cT(G)$.

Cameron, Grubb, and Sawada~\cite{MR4756593} introduced the stronger notion of a \defi{pivot}-exchange, which is an exchange~$\{e,f\}$ with the additional property that~$e$ and~$f$ have a common end vertex.
For example, in the graph~$G$ shown in Figure~\ref{fig:span}, all exchanges except~$\{1,5\}$ and~$\{2,4\}$ are pivot-exchanges.
They raised the following problem.

\begin{problem}
\label{prob:vp}
Does every graph~$G$ admit a pivot-exchange Gray code of its spanning trees?
\end{problem}

\begin{figure}[t!]
\includegraphics[page=2]{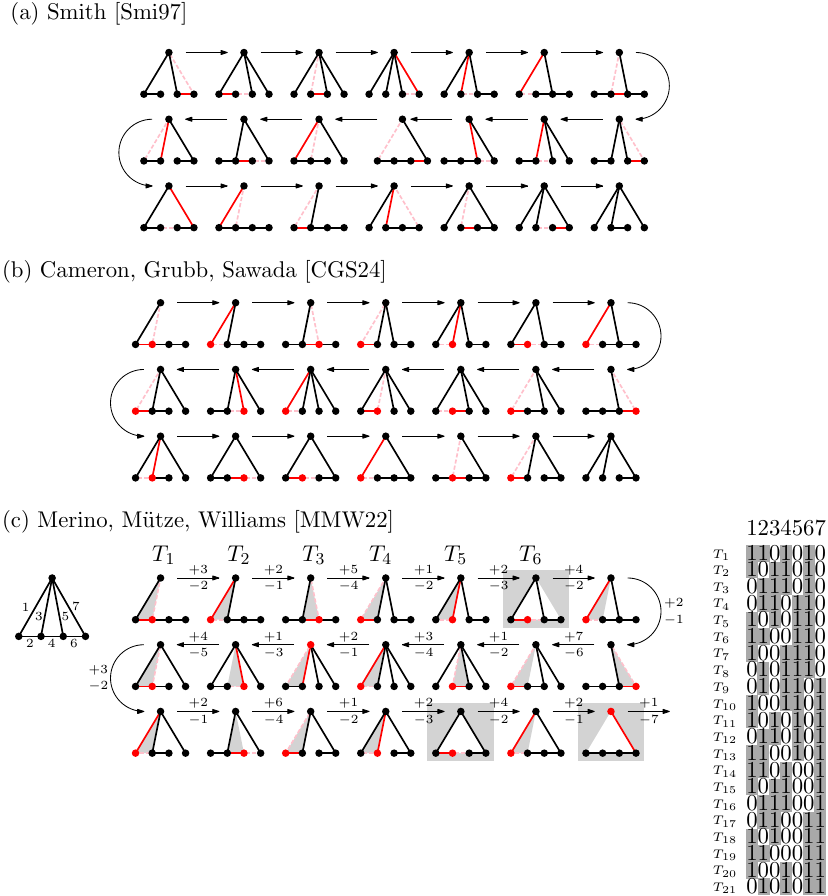}
\caption{Three different edge exchange Gray codes for listing the 21 spanning trees of the fan graph~$F_5$.
In each spanning tree, the edge removed to reach the next tree is highlighted (prefixed by~$-$ in~(c)), and the non-edge being added is dashed (prefixed by~$+$ in~(c)).
In (b) and~(c), the common end vertex of each pivot-exchange operation is highlighted, and in~(c), the common face of each face-exchange operation is highlighted.
The right-hand side of~(c) shows the characteristic vectors of each spanning tree.}
\label{fig:3GCs}
\end{figure}

Additionally, they asked if such a listing can be computed using a greedy strategy, and possibly by an efficient algorithm.

Problem~\ref{prob:vp} is a special case of a more general question raised by Knuth in Vol.~4A of his seminal series~`The Art of Computer Programming'~\cite{MR3444818}, stated as problem~102 in Section~7.2.1.6 with a difficulty rating of 46/50:
\footnote{A flawed attempt at settling this problem was published in~\cite{MR263701} (see~\cite{rao_raju_1972}).}

\begin{problem}
\label{prob:dir}
Does every directed graph admit an edge exchange Gray code of its \defi{oriented spanning trees}, also known as \defi{arborescences}, i.e., spanning trees in which all arcs are oriented away from a fixed root vertex~$r$?
\end{problem}

Note that for any exchange~$\{e,f\}$ in this directed setting, in order to preserve the arborescence property, the arcs~$e$ and~$f$ must point to the same vertex, i.e., it must be a pivot-exchange.
Furthermore, a positive answer to Problem~\ref{prob:dir} would imply an affirmative answer to Problem~\ref{prob:vp}:
Indeed, given an undirected graph~$G$, construct the directed graph~$G^\leftrightarrows$ by replacing each undirected edge of~$G$ by two oppositely oriented arcs, and pick an arbitrary vertex as root~$r$.
Listing the oriented spanning trees of~$G^\leftrightarrows$ by arc exchanges produces every spanning tree of~$G$ exactly once, namely in the orientation forced by the choice of~$r$.

\subsection{Fan graphs}

As a first step towards Problem~\ref{prob:vp}, Cameron, Grubb, and Sawada~\cite{MR4756593} provided a pivot-exchange Gray code for listing the spanning trees of \defi{fan graphs~$F_n$}, which are obtained by joining an extra vertex to all vertices of a path on~$n-1$ vertices; see Figure~\ref{fig:fan}.

\begin{figure}[h!]
\includegraphics[page=1]{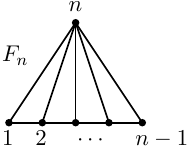}
\caption{The fan graphs~$F_n$.}
\label{fig:fan}
\end{figure}

\begin{theorem}[{\cite[Thm.~4]{MR4756593}}]
\label{thm:fan1}
For any $n\geq 3$, there is a pivot-exchange Gray code for the spanning trees of~$F_n$.
\end{theorem}

The output of their algorithm for the case~$n=5$ is shown in Figure~\ref{fig:3GCs}~(b).
The algorithm uses a greedy strategy that prioritizes exchanges based on vertex labels.

\subsection{A simple greedy algorithm}

Merino, M\"utze and Williams~\cite{MR4473269} discovered a simple greedy algorithm for listing the spanning trees of a graph~$G$ by (arbitrary) edge exchanges.
The algorithm operates based on a total ordering of the edges of~$G$, which is captured by labeling the edges by integers.
Specifically, if $m$ denotes the number of edges of~$G$, then an \defi{edge labeling} is a bijection~$\ell:E(G)\rightarrow \{1,\ldots,m\}$.
For an edge~$e\in E(G)$, we refer to $\ell(e)$ as the \defi{label} of the edge~$e$.
In the following examples, we will often identify edges by their labels.
In particular, edge exchanges $\{e,f\}$ will be denoted by the pairs of labels $\{\ell(e),\ell(f)\}$.
We will also use the abbreviation $[m]:=\{1,\ldots,m\}$.

\begin{algo}{Algorithm~G}{Greedy edge exchanges}
This algorithm greedily generates the spanning trees of a graph~$G$ with $m$ edges via edge exchanges, using an edge labeling $\ell:E(G)\rightarrow[m]$ and an initial spanning tree~$\tT\in\cT(G)$.
\begin{enumerate}[label={\bfseries G\arabic*.}, leftmargin=10mm, noitemsep, topsep=3pt plus 3pt]
\item{} [Initialize] Visit the initial spanning tree~$\tT$.
\item{} [Exchange] Perform an edge exchange in the current spanning tree that minimizes the larger of the two edge labels in the exchange and that yields an unvisited spanning tree from~$\cT(G)$.
If no such exchange exists, then terminate.
Otherwise visit this new spanning tree and repeat~G2.
\end{enumerate}
\end{algo}

The output of Algorithm~G when applied to the fan graph~$F_5$ is shown in Figure~\ref{fig:3GCs}~(c), using the edge labeling displayed on the left.
To illustrate the greedy rule in step~G2, when reaching the sixth spanning tree~$T_6=\{1,2,5,6\}$, there are seven possible edge exchanges to obtain another spanning tree in~$\cT(F_5)$, namely $\{1,3\}$, $\{2,3\}$, $\{1,4\}$, $\{2,4\}$, $\{4,5\}$, $\{5,7\}$ and $\{6,7\}$; see Figure~\ref{fig:algoG}.
Only $\{1,4\}$, $\{2,4\}$, $\{5,7\}$ and $\{6,7\}$ give an unvisited spanning tree, and among those, $\{1,4\}$ and~$\{2,4\}$, minimize the larger label, which is~4.
In this case, the exchange~$\{2,4\}$ is applied, but $\{1,4\}$ would be a valid alternative for the algorithm.

\begin{figure}[h!]
\includegraphics[page=3]{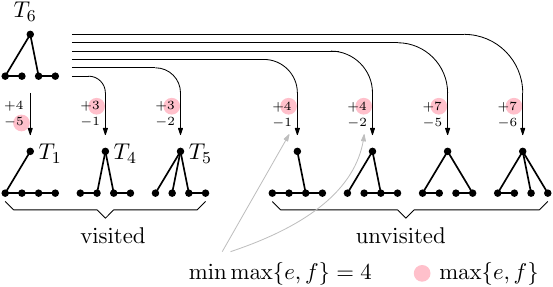}
\caption{Illustration of the sixth iteration of Algorithm~G in the run shown in Figure~\ref{fig:3GCs}~(c).}
\label{fig:algoG}
\end{figure}

In general, in step~G2 there may be several edge exchanges
\begin{equation}
\label{eq:tie}
\{e_1,f\},\{e_2,f\},\ldots,\{e_t,f\} \text{ with  } \ell(e_1)<\ell(e_2)<\cdots<\ell(e_t)<\ell(f),
\end{equation}
applicable to the current spanning tree to give an unvisited spanning tree, which all have $\ell(f)$ as the larger label, differing only in the smaller label~$\ell(e_i)$, $i\in[t]$.
We refer to such a situation as a \defi{tie}.
A \defi{tie-breaking rule} is a procedure that determines which exchange to apply in case of a tie.

By definition of the step~G2, Algorithm~G only selects edge exchanges that result in a previously unvisited spanning tree of~$G$, i.e., the produced listing of spanning trees will not contain repetitions.
However, it could be that the algorithm terminates before having visited the entire set~$\cT(G)$, a situation that is ruled out by the next theorem.

\begin{theorem}[{\cite[Thm.~10]{MR4473269}}]
\label{thm:greedy}
For any graph~$G$, for any edge labeling of~$G$, for any initial spanning tree~$\tT\in\cT(G)$, and for any tie-breaking rule, Algorithm~\upright{G} yields a genlex listing of all spanning trees of~$G$.
\end{theorem}

A listing of bitstrings is called \defi{genlex} if all strings with the same suffix appear consecutively.
In particular, all strings ending with~0 appear before all strings ending with~1, or vice versa.
A genlex listing of bitstrings is also sometimes referred to as suffix-partitioned in the literature.
Note that genlex order generalizes colexicographic order.
In the context of spanning trees, the \defi{genlex} property refers to the corresponding characteristic vectors~$\chi(T)$ of spanning trees~$T\in\cT(G)$; see the right-hand side of Figure~\ref{fig:3GCs}~(c).
In other words, all spanning trees not containing the highest-labeled edge appear before all spanning trees containing this edge, or vice versa, and this property is true recursively within the blocks.

As evidenced by this theorem, the greedy Algorithm~G is very powerful and versatile.
In fact, the algorithm can be generalized for listing the bases of any matroid by base exchanges, and even more generally, for traversing a Hamilton path on the skeleton of any 0/1-polytope~\cite{MR4795009}.

\subsection{The face-exchange property}

The paper~\cite{MR4473269} also introduced another closeness condition for edge exchanges, which is well-defined only for plane graphs.
Specifically, a \defi{face}-exchange between two spanning trees of a plane graph is an exchange of two edges that lie on a common face.
If an edge exchange is both a pivot-exchange \emph{and} face-exchange, then we refer to it as a \defi{\paf{}}-exchange.
A weaker requirement is that it is a pivot-exchange \emph{or} a face-exchange, and then we refer to it as a \defi{\pof{}}-exchange.
These notions give rise to the following question.

\begin{problem}
\label{prob:vfp}
Does every plane graph~$G$ admit a \paf{}-exchange or a \pof{}-exchange Gray code of its spanning trees?
\end{problem}

Pivot- and face-exchanges are connected through the well-known concept of the dual graph; see Figure~\ref{fig:dual}~(a).
For a plane graph~$G$, we write~$F(G)$ for the set of faces of~$G$.
The \defi{dual graph}~$G'$ is the plane graph obtained from~$G$ as follows:
For every face~$\alpha\in F(G)$, the dual graph~$G'$ has a vertex~$\alpha'\in V(G')$, and for every edge~$e$ of~$G$ between faces~$\alpha$ and~$\beta$ of~$G$, the dual graph~$G'$ has the edge~$e'=(\alpha',\beta')$.
For every vertex~$v\in V(G)$, we write $v'\in F(G')$ for the corresponding face of~$G'$ dual to it.
For a spanning tree~$T\in\cT(G)$, the \defi{dual spanning tree~$T'\in\cT(G')$} has the dual edge~$e'\in E(T')$ for every non-edge~$e\in E(G)\setminus E(T)$ and a dual non-edge $e'\notin E(T')$ for every edge $e\in E(T)$; see Figure~\ref{fig:dual}~(b).
Observe that a pivot-exchange~$\{e,f\}$ in~$T$ is a face-exchange~$\{e',f'\}$ in the dual spanning tree~$T'$; see Figure~\ref{fig:dual}~(c).
Symmetrically, a face-exchange~$\{e,f\}$ in~$T$ is a pivot-exchange~$\{e',f'\}$ in~$T'$.

\begin{figure}[h!]
\includegraphics[page=1]{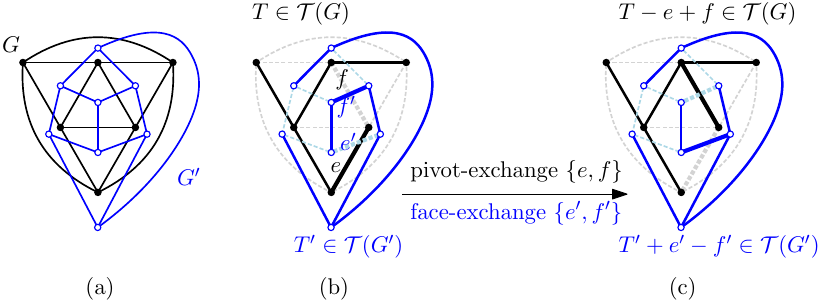}
\caption{Connection between pivot- and face-exchanges via the dual spanning tree.}
\label{fig:dual}
\end{figure}

Merino, M\"utze and Williams~\cite{MR4473269} strengthened Theorem~\ref{thm:fan1}, by showing that for a suitable edge labeling of the fan graph~$F_n$ and a suitable tie-breaking rule, Algorithm~G yields a \paf{}-exchange Gray code.
Specifically, the edges of~$F_n$ are labeled from \defi{left-to-right}, as shown in Figure~\ref{fig:3GCs}~(c), and ties are broken according to the \defi{closest} tie-breaking rule, which in case of a tie as in~\eqref{eq:tie} selects the exchange~$\{e_t,f\}$, i.e., the exchange that maximizes the smaller of the edge labels~$\ell(e_t)$ (equivalently, the one for which the smaller label is closest to the larger label~$\ell(f)$).

\begin{theorem}
\label{thm:fan2}
For any $n\geq 3$, for the left-to-right labeling of the edges of~$F_n$, for any initial spanning tree~$\tT\in\cT(F_n)$, and for the closest tie-breaking rule, Algorithm~G yields a genlex \paf{}-exchange Gray code for the spanning trees of~$F_n$.
\end{theorem}

In fact, the Gray code mentioned in Theorem~\ref{thm:fan1} also has the stronger \paf{}-exchange property.

\subsection{Our results}

In this work, we solve Problems~\ref{prob:vp} and~\ref{prob:vfp} for certain families of plane graphs that generalize fan graphs, thus generalizing Theorems~\ref{thm:fan1} and~\ref{thm:fan2}.

An \defi{outerplane} graph is a plane graph in which all vertices are incident to the outer face.
An outerplane graph is a \defi{triangulation} if all of its faces, except possibly the outer face, are triangles.
Note that fan graphs are a very special case of outerplane triangulations.

\begin{figure}
\begin{tabular}{ccccc}
\raisebox{-\height}{\includegraphics[page=8]{dual}} &
\raisebox{-\height}{\includegraphics{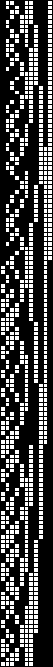}} & \hspace{10mm} &
\raisebox{-\height}{\includegraphics[page=9]{dual}} &
\raisebox{-\height}{\includegraphics{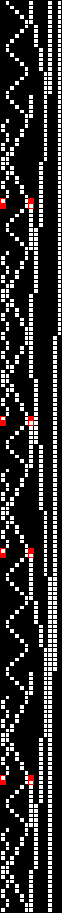}} \\
\raisebox{5mm}{} & (a) & & & (b)
\end{tabular}
\caption{(a) Pivot-exchange Gray code for the spanning trees of the outerplane triangulation~$G$.
(b) \Pof{}-exchange Gray code for the spanning trees of the outerplane graph~$G'$, which has the four marked face-exchanges~$\{1,7\}$, and the rest pivot-exchanges.
Both listings were computed by Algorithm~G, and the spanning trees are represented by their characteristic vectors, with 1-bits and 0-bits drawn as black and white squares, respectively.
The initial spanning tree is highlighted in both graphs.}
\label{fig:ex}
\end{figure}

\begin{theorem}
\label{thm:triang}
For any outerplane triangulation~$G$, there is an edge labeling of~$G$, so that for any initial spanning tree~$\tT\in\cT(G)$, there is a tie-breaking rule for which Algorithm~\upright{G} yields a genlex pivot-exchange Gray code for the spanning trees of~$G$.
\end{theorem}

\begin{theorem}
\label{thm:outerplane}
For any outerplane graph~$G$, there is an edge labeling of~$G$, so that for any initial spanning tree~$\tT\in\cT(G)$, there is a tie-breaking rule for which Algorithm~\upright{G} yields a genlex \pof{}-exchange Gray code for the spanning trees of~$G$.
\end{theorem}


These theorems directly yield efficient algorithms.
Specifically, using the techniques described in~\cite[Sec.~7.2+Cor.~32]{MR4795009}, Algorithm~G can be implemented to output each spanning tree in time~$\cO(n\log n)$, where $n$ is the number of vertices of~$G$.
The required space for the algorithm is~$\cO(n)$.
In particular, this implementation is \defi{history-free} in the sense that no previously computed spanning trees apart from the current spanning tree need to be stored, plus some simple data structures for bookkeeping.

Two examples of Gray code listings of spanning trees obtained from these theorems are shown in Figure~\ref{fig:ex}.

\begin{corollary}
\label{cor:outerplane}
Any outerplane triangulation admits a genlex pivot-exchange Gray code of its spanning trees.
Any outerplane graph admits a genlex \pof{}-exchange Gray code of its spanning trees.
\end{corollary}

The \defi{weak dual} graph of a plane graph~$G$, denoted~$G^-$ is the graph obtained from the dual graph~$G'$ by removing the vertex corresponding to the outer face of~$G$; see Figure~\ref{fig:weak-dual}~(a).

\begin{figure}[h!]
\includegraphics[page=2]{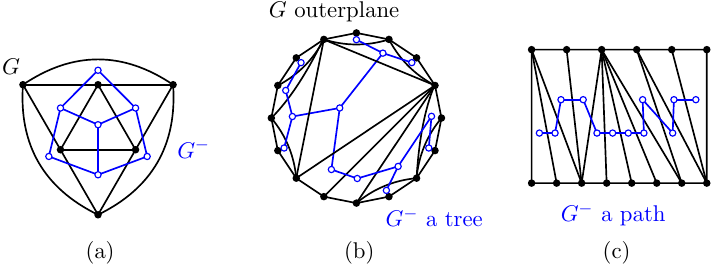}
\caption{Illustration of (a) weak dual graph; (b) 2-connected outerplane triangulation and its weak dual tree; (e) 2-connected outerplane triangulation whose weak dual tree is a path.}
\label{fig:weak-dual}
\end{figure}

Note that~$G$ is a 2-connected outerplane graph if and only if the weak dual~$G^-$ is a tree; see Figure~\ref{fig:weak-dual}~(b).
Also note that $G^-$ is a path if and only if $G$ is 2-connected and all but at most two sides of every inner face are incident with the outer face; see Figure~\ref{fig:weak-dual}~(c).
In particular, for a triangulation~$G$, we have that $G^-$ is a path if and only if $G$ is 2-connected and at least one side of every triangle touches the outer face.
This is true in particular for fan graphs~$F_n$.

We show that fan graphs, and more generally outerplane triangulations for which the weak dual is a path, are exactly the outerplane graphs that have the maximum number of spanning trees for a fixed number of edges.
Moreover, the counts are the Fibonacci numbers, defined as $f_0:=0$, $f_1:=1$ and
\begin{equation}
\label{eq:fib}
f_{m+1}:=f_m+f_{m-1}
\end{equation}
for all $m\geq 1$.
We write $t(G):=|\cT(G)|$ for the number of spanning trees of~$G$.

\begin{theorem}
\label{thm:fib}
For any outerplane graph~$G$ with $m$ edges we have $t(G)\leq f_{m+1}$, with equality if and only if $G$ is a triangulation such that $G^-$ is a path.
\end{theorem}

The identity~$t(G)=f_{m+1}$ when~$G$ is a triangulation for which $G^-$ is a path was already noticed by Slater~\cite[Prop.~1]{MR437348}.

\begin{remark}
\label{rem:multi}
We remark that Problems~\ref{prob:vp}, \ref{prob:dir} and~\ref{prob:vfp} are perfectly valid also for multigraphs, i.e., graphs that may have parallel edges and/or loops (though loops are irrelevant in the context of spanning trees).
In fact, Theorems~\ref{thm:greedy}, \ref{thm:triang} and~\ref{thm:outerplane}, and hence Corollary~\ref{cor:outerplane}, also hold in this more general setting.
An `outerplane triangulation' in this case has all inner faces of lengths at most~3, instead of exactly~3.
The time and space bounds stated after Theorem~\ref{thm:outerplane} change to $\cO(m\log n)$ and~$\cO(m)$, respectively, where $m$ is the number of edges of the multigraph.
However, for simplicity we do our proofs in the setting of simple graphs, where no parallel edges nor loops are allowed.
Nonetheless, our proof of Theorem~\ref{thm:fib} will actually use multigraphs.
\end{remark}

\begin{remark}
The listings of spanning trees produced by Algorithm~$G$ are in general not Hamilton cycles in~$\cF(G)$, but only Hamilton paths, i.e., the first and last spanning tree in general do not differ in an edge exchange.
This remark applies to all the Gray codes mentioned in Theorems~\ref{thm:greedy}, \ref{thm:fan2}, \ref{thm:triang}, and~\ref{thm:outerplane}, and in Corollary~\ref{cor:outerplane}.
For some graphs, some edge labelings, and some initial spanning trees, however, the Gray codes are cyclic, such as the one mentioned in Theorem~\ref{thm:fan2} for the L-shaped initial spanning tree shown in Figure~\ref{fig:3GCs}~(c), which ends with the mirrored~L.
\end{remark}

\subsection{Related work}

There has been an extensive amount of work~\cite{MR401486,MR495152,MR1453415,MR1146706,MR1355446,MR1448628} on efficiently generating the set~$\cT(G)$ of all spanning trees of~$G$, \emph{without} the requirement that any two consecutive trees differ in an edge exchange, i.e., the computed listings are \emph{not} Hamilton paths or cycles in the flip graph~$\cF(G)$.
The algorithms in the last three of the aforementioned papers achieve this in time~$\cO(1)$ on average per generated tree.
See the survey~\cite{chakraborty_et_al_2019} for a comparison of the different algorithms.

If instead of listing all spanning trees of~$G$, we want to count them, then this can be achieved by Kirchhoff's Matrix-Tree Theorem, which reduces the problem to computing a determinant, which can be done efficiently.
This problem is closely connected to finding the so-called \defi{most vital edge} of~$G$, which is the edge contained in the most spanning trees from~$\cT(G)$.
Random sampling~\cite{DBLP:conf/focs/Broder89,MR1069105} and ranking/unranking~\cite{MR993191,MR1379225} of spanning trees have also been considered.

\subsection{Outline of this paper}

In Section~\ref{sec:listing} we prove Theorems~\ref{thm:triang} and~\ref{thm:outerplane}.
In Section~\ref{sec:fib} we prove Theorem~\ref{thm:fib}.
We conclude with some open questions in Section~\ref{sec:open}, and there we also report on some experimental evidence.

\section{Proof of Theorems~\ref{thm:triang} and~\ref{thm:outerplane}}
\label{sec:listing}

For the proofs of Theorems~\ref{thm:triang} and~\ref{thm:outerplane}, we will assume w.l.o.g. that the outerplane graph~$G$ is 2-connected.
If $G$ is not 2-connected, then all the arguments presented in the following apply to each of its blocks, and the spanning trees of~$G$ are obtained by combining the spanning trees in each block in all possible ways, and pivot- or face-exchanges remain valid within the blocks.

We first describe the labeling of edges of a 2-connected outerplane graph~$G$ that we use in order to run Algorithm~G on the graph~$G$.
We then establish two important properties of these labelings (Lemmas~\ref{lem:labels} and~\ref{lem:neighbors}) that are crucial to show that whenever an arbitrary edge exchange becomes applicable in step~G2 of Algorithm~G, then ties can be broken to instead use a pivot- or face-exchange (Lemma~\ref{lem:pivot}).

\subsection{The edge labeling}

To define the edge labeling, we need another definition.
Let $v$ be the vertex of~$G'$ corresponding to the outer face of~$G$, and let $d$ be the degree of~$v$.
The \defi{split dual} graph, denoted~$G^*$, is obtained from~$G'$ by splitting~$v$ into $d$ many degree-1 vertices, one adjacent to each neighbor of~$v$; see Figure~\ref{fig:split-dual}~(a).

Note that the weak dual graph~$G^-$ is a tree if and only if the split dual graph~$G^*$ is a tree; see Figures~\ref{fig:weak-dual}~(b) and~\ref{fig:split-dual}~(b).
Furthermore, $G^-$ is a path if and only if $G^*$ is a caterpillar, i.e., a tree in which all vertices are in distance~$\leq 1$ from a central path; see Figures~\ref{fig:weak-dual}~(c) and~\ref{fig:split-dual}~(c).

\begin{figure}[h!]
\includegraphics[page=3]{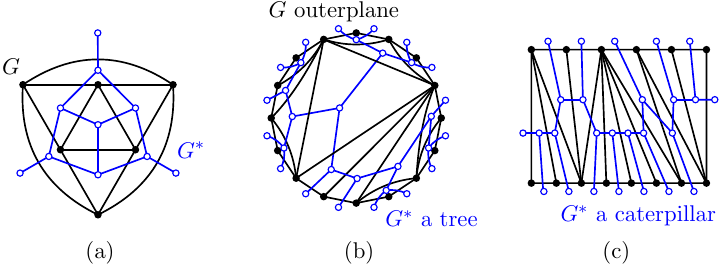}
\caption{Illustration of (a) split dual graph; (b) 2-connected outerplane triangulation and its split dual tree; (c) 2-connected outerplane triangulation whose split dual tree is a caterpillar (cf. Figure~\ref{fig:weak-dual}).}
\label{fig:split-dual}
\end{figure}

Let $G$ be a 2-connected outerplane graph with $m$ edges.
We consider the split dual tree~$G^*$, we pick a leaf~$r$ of this tree as root, and we orient all of its edges away from~$r$, giving an oriented tree~$\rvec{G^*}$, referred to as \defi{oriented split dual}.
We label the edges of~$\rvec{G^*}$ in a depth-first-search manner with integers~$1,\ldots,m$, starting at the root~$r$ and processing subtrees in counterclockwise (ccw) order; see Figure~\ref{fig:label}.
The labeling of the edges of~$\rvec{G^*}$ induces a labeling of the corresponding dual edges of~$G$ with integers from~$[m]$.
We refer to this labeling of the edges of~$G$ as \defi{dual-tree labeling}.
Note that there is freedom in the choice of the root in the tree~$G^*$, and different choices yield different oriented split dual trees~$\rvec{G^*}$ and hence different edge labelings of the same graph~$G$.
Note that the left-to-right labeling of the fan graph~$G=F_n$ shown in Figure~\ref{fig:3GCs}~(c) is one particular dual-tree labeling.

\begin{figure}[h!]
\includegraphics[page=4]{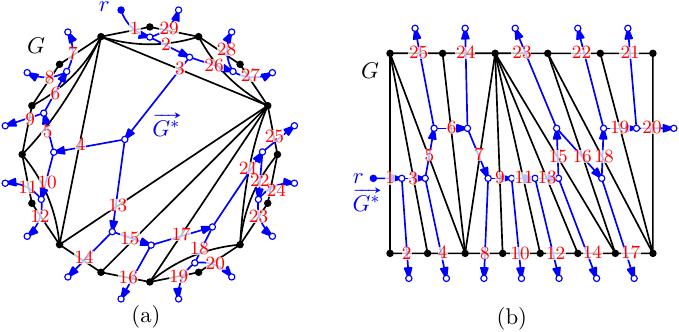}
\caption{Illustration of the dual-tree labeling procedure.}
\label{fig:label}
\end{figure}

\subsection{Proofs of theorems}

Throughout this section, we let $G$ be a 2-connected outerplane graph with $m$ edges, $\rvec{G^*}$ an oriented split dual tree, and $\ell:E(G)\rightarrow[m]$ the corresponding dual-tree labeling of the edges of~$G$.
Before proving Theorems~\ref{thm:triang} and~\ref{thm:outerplane}, we first derive two properties of the edge labelings defined in the previous section, stated in Lemmas~\ref{lem:labels} and~\ref{lem:neighbors} below.

\begin{figure}[h!]
\includegraphics[page=5]{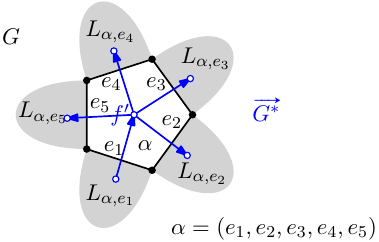}
\caption{Illustration of lobes.}
\label{fig:lobe}
\end{figure}

The following definitions are illustrated in Figure~\ref{fig:lobe}.
We denote a face~$\alpha$ of~$G$ of length~$t$ by the sequence of edges~$(e_1,\ldots,e_t)$ bounding this face in ccw order, starting with the edge~$e_1$ whose dual edge in~$\rvec{G^*}$ is oriented towards~$\alpha'$ (all other edges of~$\rvec{G^*}$ are oriented away from~$\alpha'$).
For a face~$\alpha=(e_1,\ldots,e_t)$ and an index~$i\in[t]$, the \defi{$(\alpha,e_i)$-lobe $L_{\alpha,e_i}$} is the set of edges of~$G$ that are dual to the edges in the maximal subtree of~$G^*$ that contains the edge~$e_i'$, but none of the other edges $e_j'$, $j\in[t]\setminus \{i\}$.

\begin{lemma}
\label{lem:labels}
For any face $\alpha=(e_1,\ldots,e_t)$ of~$G$, we have $\ell(e_1)<\ell(e_2)<\cdots<\ell(e_t)$.
Furthermore, for any $i\in\{2,\ldots,t\}$ and $f\in L_{\alpha,e_i}\setminus\{e_i\}$ we have $\ell(e_i)<\ell(f)$, and $\ell(f)<\ell(e_{i+1})$ if $i<t$.
\end{lemma}

\begin{proof}
This is an immediate consequence of the labeling procedure which processes subtrees in the oriented split dual in ccw order and in a depth-first-search manner.
Specifically, for all $i\in\{2,\ldots,t\}$, all edges of~$L_{\alpha,e_i}\setminus\{e_i\}$ are labeled directly after~$e_i$, and directly before~$e_{i+1}$ if $i<t$.
\end{proof}

For a vertex~$v$ of~$G$, the \defi{incidence list~$E_v$} is the sequence of edges incident to~$v$ in clockwise order around~$v$, starting and ending with the two edges that bound the outer face; see Figure~\ref{fig:neighbors}.
We say that an edge~$e$ incident with~$v$ is a \defi{cw-edge} or \defi{ccw-edge}, respectively, according to the clockwise or counterclockwise orientation of the dual edge~$e'$ in~$\rvec{G^*}$ around~$v$.

\begin{lemma}
\label{lem:neighbors}
For any vertex~$v$, let $E_v=:(e_1,\ldots,e_t)$ be its incidence list.
Then there is an index~$i\in\{0,\ldots,t\}$ such that $e_1,\ldots,e_i$ are ccw-edges, $e_{i+1},\ldots,e_t$ are cw-edges, and we have $\ell(e_i)<\ell(e_{i-1})<\cdots<\ell(e_1)<\ell(e_{i+1})<\ell(e_{i+2})<\cdots<\ell(e_t)$.
In particular, for any cw-edge~$e_j$ and any~$i\in[j-1]$ we have $\ell(e_i)<\ell(e_j)$.
\end{lemma}

\begin{figure}[h!]
\includegraphics[page=6]{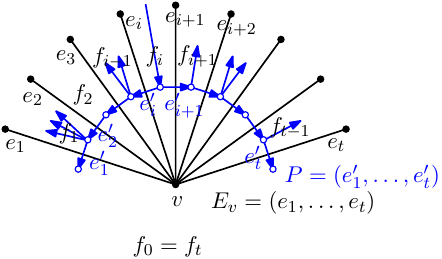}
\caption{Illustration of Lemma~\ref{lem:neighbors}.}
\label{fig:neighbors}
\end{figure}

\begin{proof}
For any sequence of edges~$P=(a_1,\ldots,a_t)$ that form a path in the oriented tree~$\rvec{G^*}$, there is an index~$i\in\{0,\ldots,t\}$ such that $a_1,\ldots,a_i$ are oriented towards the start vertex of~$P$ and $a_{i+1},\ldots,a_t$ are oriented towards the end vertex of~$P$.
The special cases~$i=0$ and~$i=t$ mean that the path is oriented conformly, towards the end or start vertex, respectively.

The sequence of dual edges~$P:=(e_1',\ldots,e_t')$ is a path in~$\rvec{G^*}$ to which the aforementioned observation applies.
For $j\in\{1,\ldots,t-1\}$, let $f_j$ be the face of~$G$ bounded by~$e_j$ and~$e_{j+1}$.
Furthermore, we denote the outer face of~$G$ by~$f_0=f_t$.
Note that $e_j'$ is oriented towards~$f_{j-1}'$ in~$\rvec{G^*}$ for $j=1,\ldots,i$, i.e., these are all ccw-edges w.r.t.~$v$.
Furthermore, $e_j'$ is oriented towards~$f_j'$ in~$\rvec{G^*}$ for $j=i+1,\ldots,t$, i.e., these are all cw-edges w.r.t.~$v$.

Applying Lemma~\ref{lem:labels} to the faces $f_j$, $j=1,\ldots,i-1$ and $j=i+1,\ldots,t-1$, yields $\ell(e_i)<\ell(e_{i-1})<\cdots<\ell(e_1)$ and $\ell(e_{i+1})<\ell(e_{i+2})<\cdots<\ell(e_t)$, respectively.
Furthermore, applying Lemma~\ref{lem:labels} to $L_{f_i,e_i}$ and~$L_{f_i,e_{i+1}}$ proves that $\ell(e_1)<\ell(e_{i+1})$.
Combining these inequalities proves the lemma. 
\end{proof}

\begin{lemma}
\label{lem:pivot}
Let $T\in\cT(G)$ and let~$\{e,f\}$ with $\ell(e)<\ell(f)$ be an edge exchange for~$T$.
Then there is a \pof{}-exchange~$\{d,f\}$ with $\ell(d)<\ell(f)$.
Furthermore, if $G$ is a triangulation, then $\{d,f\}$ is a pivot-exchange.
\end{lemma}

\begin{proof}
We let~$\alpha=(e_1,\ldots,e_t)$ be the face incident with~$f$ in~$G$ for which the dual edge~$f'$ in~$\rvec{G^*}$ is oriented away from~$\alpha'$.
As~$\ell(f)>\ell(e)\geq 1$ we have $\ell(f)>1$, and consequently $\alpha$ is not the outer face, but an inner face.
We have $f=e_i$ for some $i\in\{2,\ldots,t\}$.
The graph~$T\cup\{e,f\}$ contains exactly one cycle~$C$, which contains both edges~$e$ and~$f$.
As $\ell(e)<\ell(f)$, Lemma~\ref{lem:labels} implies that $e\notin L_{\alpha,f}$, and as $e\in C$, we obtain that $C$ contains edges from every lobe~$L_{\alpha,e_i}$ for all $i\in[t]$, and furthermore $e\in L_{\alpha,e_j}$ for some $j\in\{1,\ldots,i-1\}$.

We now distinguish the cases whether $f\in T$ or $f\notin T$, i.e., whether the exchange $\{e,f\}$ adds the edge~$e$ and removes~$f$, or removes~$e$ and adds~$f$, respectively.
These two cases are illustrated in Figure~\ref{fig:pivot}~(a) and~(b), respectively.

{\bf Case~(a):} $f\in T$.
Note that $\{d,f\}$ with $d:=e_j$ is a valid face-exchange for~$T$, and we have $\ell(d)<\ell(f)$ by Lemma~\ref{lem:labels}.
Furthermore, if $G$ is a triangulation, then we have $t=3$ and consequently $f=e_i$ and $d=e_j$ share an end vertex, so the exchange~$\{d,f\}$ is a pivot-exchange.

{\bf Case~(b):} $f\notin T$.
Let $d$ be the edge of~$C$ incident with~$f$ in the lobe~$L_{\alpha,e_{i-1}}$.
Note that $\{d,f\}$ is a valid pivot-exchange for~$T$.
If $i>2$, then we have $\ell(d)<\ell(f)$ by Lemma~\ref{lem:labels}.
If $i=2$, then let $v$ be the common end vertex of~$e_1=e_j$ and $e_2=e_i$, and consider the incidence list~$E_v$ of~$v$.
The edge~$e_1$ is a cw-edge and $d$ comes before it (or is equal to~$e_1$) on the list~$E_v$, and so Lemma~\ref{lem:neighbors} yields~$\ell(d)\leq \ell(e_1)$.
By Lemma~\ref{lem:labels} we have $\ell(e_1)<\ell(f)$ and therefore $\ell(d)<\ell(f)$.

\begin{figure}[h!]
\includegraphics[page=7]{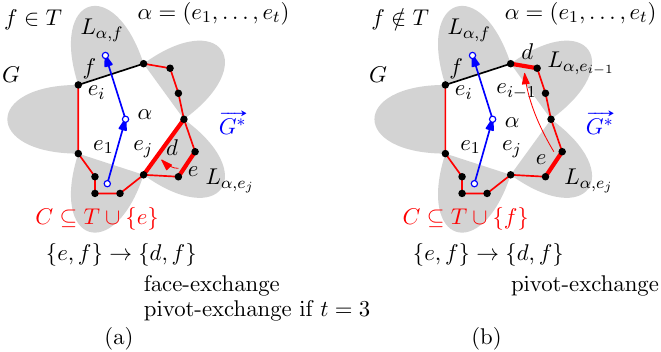}
\caption{Illustration of the two cases in the proof of Lemma~\ref{lem:pivot}.}
\label{fig:pivot}
\end{figure}
\end{proof}

We are now in position to prove Theorems~\ref{thm:triang} and~\ref{thm:outerplane}.

\begin{proof}[Proof of Theorems~\ref{thm:triang} and~\ref{thm:outerplane}]
By Lemma~\ref{lem:pivot}, whenever Algorithm~G considers an edge exchange~$\{e,f\}$ with $\ell(e)<\ell(f)$, there is a \pof{}-exchange~$\{d,f\}$ with $\ell(d)<\ell(f)$.
By Theorem~\ref{thm:greedy}, the listing of spanning trees produced by Algorithm~G has the genlex property, which implies that if $T\triangle\{e,f\}$ is unvisited, then $T\triangle\{d,f\}$ is also unvisited.
It follows that there is a tie-breaking rule for Algorithm~G to only ever use \pof{}-exchanges.

Furthermore, if $G$ is a triangulation, then by Lemma~\ref{lem:pivot} the alternative exchange~$\{d,f\}$ is a pivot-exchange, so there is a tie-breaking rule for Algorithm~G to only ever use pivot-exchanges.
\end{proof}

According to the proof, the tie-breaking rule for Algorithm~G that works is to select an arbitrary pivot-exchange or \pof{}-exchange, respectively.

\section{Proof of Theorem~\ref{thm:fib}}
\label{sec:fib}

We prove Theorem~\ref{thm:fib} in the more general setting of multigraphs (recall Remark~\ref{rem:multi}), i.e., from now on we allow multiple edges between pairs of vertices.
Two edges between the same pair of vertices are sometimes referred to as \defi{parallel} edges.
Loops may be present, but are never contained in a spanning tree and hence irrelevant for us.
However, the distinction of parallel edges is important.
For example, the plane graph formed by two vertices connected by $m$ parallel edges has $m$ different spanning trees, each containing exactly one of the $m$ edges.

All notions introduced in Section~\ref{sec:intro} are valid in the more general context of multigraphs.
The only exception is the definition of an \defi{outerplane triangulation}, which we change from `face length~3' to `face length $\leq 3$'.
A length-2 face comes from two parallel edges, and is called a \defi{digon}.

In addition to edge removal and edge addition, we now also introduce the operation of \defi{edge contraction}.
Given a graph~$G$ and an edge~$e\in E(G)$, we write $G/e$ for the graph obtained from~$G$ by contracting the edge~$e$.
Note that even if $G$ is a simple graph, $G/e$ may still contain parallel edges, and if $G$ has parallel edges, then $G/e$ may contain loops (which can be ignored when counting spanning trees).
The number of spanning trees~$t(G)$ of a graph~$G$ obeys the well-known recursive relation
\begin{equation}
\label{eq:tG}
t(G)=t(G-e) + t(G/e),
\end{equation}
valid for any (non-loop) edge~$e\in E(G)$ (see~\cite{minty_1965}).

We need the following auxiliary lemma about Fibonacci numbers.

\begin{lemma}
\label{lem:fib}
For any two integers $i,j \geq 1$, we have $f_i \cdot f_j \leq f_{i+j-1}$, with equality if and only if $i=1$ or $j=1$.
\end{lemma}

\begin{proof}
If $i=1$ or $j=1$, then one can check directly that the claimed equality holds.

Now assume that $i,j\geq 2$.
We prove that $f_i\cdot f_j < f_{i+j-1}$ by induction on $i+j$.
In the base case $i+j=4$ we have $i=j=2$ and therefore $f_i=f_j=f_2=1$ and $f_{i+j-1}=f_3=2$, so the claim is true.
For the induction step we distinguish the cases $j\in\{2,3\}$ and $j>3$.
If $j=2$ we have $f_i \cdot f_j = f_i < f_{i+1} = f_{i+j-1}$ (recall that $i\geq 2$).
If $j=3$ we have $f_i \cdot f_j = 2f_i < f_i+f_{i+1} = f_{i+2} = f_{i+j-1}$.
It remains to consider the case $j>3$.
We have
\[
f_i \cdot f_j = f_i \cdot (f_{j-2} + f_{j-1}) = f_i \cdot f_{j-2} + f_i \cdot f_{j-1} < f_{i+j-3} + f_{i+j-2} = f_{i+j-1},
\]
where the strict inequality in the third step holds by induction, using that $j-2>1$.
\end{proof}

Theorem~\ref{thm:fib} is an immediate consequence of the following more general statement for multigraphs.

\begin{theorem}
\label{thm:fib-multi}
For any outerplane (multi)graph~$G$ with $m$ edges we have $t(G)\leq f_{m+1}$, with equality if and only if $G$ is a triangulation such that $G^-$ is a path and all digons are incident with the outer face.
\end{theorem}

Note that digons incident with the outer face correspond to degree~1 vertices in the dual graph~$G^-$, so if $G^-$ is a path, then there are at most two such digons corresponding to end vertices of the path.

Theorem~\ref{thm:fib-multi} is illustrated in Figure~\ref{fig:fib}.

\begin{figure}[h!]
\includegraphics[page=1]{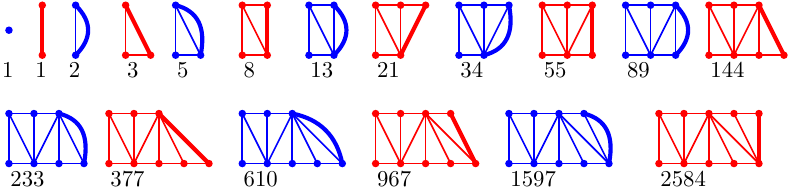}
\caption{Sequence of outerplane graphs that maximize the number of spanning trees, with the corresponding counts, the Fibonacci numbers.
They are all triangulations for which $G^-$ is a path, and they either have no digons (red) or one digon incident with the outer face (blue).
For each graph, deleting or contracting the bold edge yields the two preceding graphs.}
\label{fig:fib}
\end{figure}

\begin{proof}
We argue by induction on~$m$.
The induction basis~$m=0$ is trivial.
For the induction step we assume that~$m\geq 1$.

If $G$ is not 2-connected, then it has at least two blocks~$A$ and~$B$ with $a\geq 1$ and~$b\geq 1$ edges, respectively, where $a+b=m$.
We have $t(G)=t(A)\cdot t(B)$ and therefore $t(G)\leq f_{a+1}\cdot f_{b+1}<f_{a+b+1}=f_{m+1}$, where the strict inequality follows from Lemma~\ref{lem:fib}.

It remains to consider the case that $G$ is 2-connected.
If $G$ is a single edge, then $t(G)=1=f_2$.
If $G$ has only one inner face, then it is a cycle and we have $t(G)=m\leq f_{m+1}$, and this inequality is tight for $m\in\{2,3\}$ and strict for $m\geq 4$.
The same estimates hold if all inner faces of~$G$ are digons, i.e., $G$ has 2 vertices and all $m$ edges are parallel to each other.
For the rest of the proof we assume that $G$ has at least two faces, not all of which are digons.
This implies in particular that $m\geq 4$.
We consider a face~$\alpha$ of~$G$ for which the dual vertex~$\alpha'$ has degree at most~1 in~$G^-$.
We denote the sequence of edges bounding~$\alpha$ by $e_1,\ldots,e_t$ in clockwise order, ending with the edge~$e_t$ that has a dual edge in~$G^-$.
We distinguish two cases, namely whether $\alpha$ is a digon ($t=2$) or not ($t\geq 3$), and we bound the number of spanning trees using~\eqref{eq:tG} with respect to the edge~$e:=e_1$, i.e., the first edge of the face~$\alpha$.
These two cases are illustrated in Figure~\ref{fig:fib-proof}~(a) and~(b), respectively.

{\bf Case~(a):} $t=2$.
$G-e$ has $m-1$ edges, and therefore
\begin{equation}
\label{eq:tGmea}
t(G-e)\leq f_m
\end{equation}
by induction.
Let $s\geq 2$ be the number of edges that are parallel to~$e_1$ (including~$e_1$ and~$e_2$), and denote them by $e_1,e_2,\ldots,e_s$ in ccw order around the corresponding common end vertex; see Figure~\ref{fig:fib-proof}~(a).
Furthermore, we let $\beta$ be the face incident with~$e_s$ but not~$e_{s-1}$.
In the graph~$G/e$, the edges~$e_2,\ldots,e_s$ are loops, and we can therefore remove all of them, so we have
\begin{equation}
\label{eq:tGcea}
t(G/e)=t(G/e-\{e_2,\ldots,e_s\})\leq f_{m-1}
\end{equation}
by induction.
Combining \eqref{eq:fib}, \eqref{eq:tG}, \eqref{eq:tGmea}, and~\eqref{eq:tGcea} proves that $t(G)\leq f_{m+1}$, as claimed.
Furthermore, this inequality is tight if and only if~\eqref{eq:tGmea} and~\eqref{eq:tGcea} are tight, which happens if and only if the conditions stated in the theorem hold for $G-e$ and~$G/e$, which happens if and only if $s=2$ and $\beta$ is a triangle that is incident with the outer face, if and only if the conditions stated in the theorem hold for~$G$.

\begin{figure}[h!]
\includegraphics[page=2]{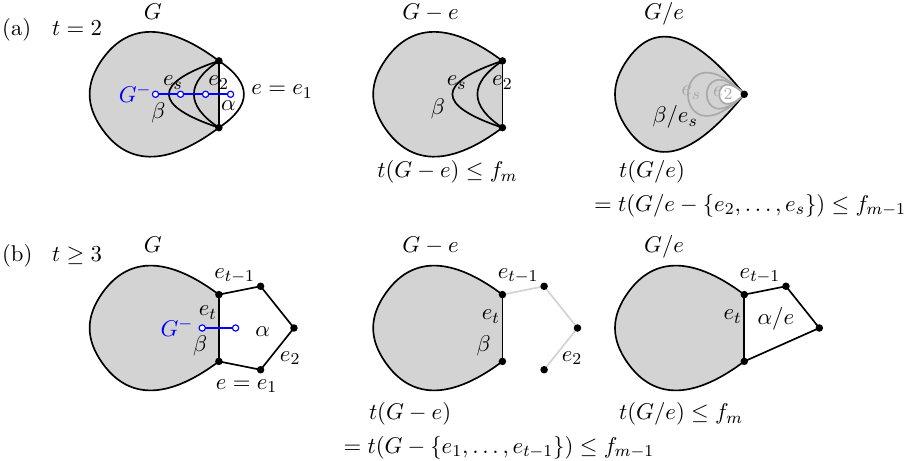}
\caption{Illustration of the proof of Theorem~\ref{thm:fib-multi}.}
\label{fig:fib-proof}
\end{figure}

{\bf Case~(b):} $t\geq 3$.
Let $\beta$ be the face incident with~$e_t$ but not~$e_1$.
$G/e$ has $m-1$ edges, and therefore
\begin{equation}
\label{eq:tGceb}
t(G/e)\leq f_m
\end{equation}
by induction.
In the graph~$G-e$, the edges $e_2,\ldots,e_{t-1}$ are present in every spanning tree, so we have 
\begin{equation}
\label{eq:tGmeb}
t(G-e)=t(G-\{e_1,\ldots,e_{t-1}\})\leq f_{m-1} 
\end{equation}
by induction.
Combining \eqref{eq:fib}, \eqref{eq:tG}, \eqref{eq:tGceb} and~\eqref{eq:tGmeb} proves that $t(G)\leq f_{m+1}$, as claimed.
Furthermore, this inequality is tight if and only if~\eqref{eq:tGceb} and~\eqref{eq:tGmeb} are tight, which happens if and only if the conditions stated in the theorem hold for $G/e$ and~$G-e$, which happens if and only if $t=3$ and $\beta$ is a triangle that is incident with the outer face, if and only if the conditions stated in the theorem hold for~$G$.
\end{proof}

\section{Open problems}
\label{sec:open}

Problems~\ref{prob:vp} and~\ref{prob:vfp} remain open for general graphs, i.e., for graphs not covered by Theorems~\ref{thm:triang} and~\ref{thm:outerplane} (recall Remark~\ref{rem:multi}).
Also, Knuth's more general Problem~\ref{prob:dir} for directed graphs remains very much open.

Concerning Problem~\ref{prob:vp}, we checked experimentally that all 2-connected graphs on up to 6 vertices admit a pivot-exchange Gray code for their spanning trees, and these Gray codes are all cyclic.
Regarding Problem~\ref{prob:dir}, we checked that all directed graphs on up to 5 vertices (oppositely oriented arcs are allowed) whose underlying undirected graph is 2-connected admit an edge exchange Gray code for their arborescences, for any choice of fixed root.
In some cases those flip graphs admit no Hamilton cycle, but only a Hamilton path (see~\cite{rao_raju_1972}), so not all of these Gray codes are cyclic.
Regarding Problem~\ref{prob:vfp}, we checked that all outerplane graphs on up to 6 vertices admit a cyclic \paf{}-exchange Gray code for their spanning trees.

Furthermore, Knuth~\cite{MR3444818} asked in Exercise~101 in Section~7.2.1.6 whether the complete graph~$K_n$ admits a `nice' Gray code listing of its spanning trees.
One interpretation of this question, in the spirit of Problem~\ref{prob:vp}, could be: Is there a pivot-exchange Gray code for the spanning trees of~$K_n$?
Alternatively, is there a Gray code listing that provides a more fine-grained explanation of Cayley's formula~$n^{n-2}$ for the number of spanning trees?
Similar questions can be asked for the complete bipartite graph~$K_{n,n}$.

Does Theorem~\ref{thm:fib} hold without the outerplane requirement for $m$ sufficiently large?
The only complete graphs~$K_n$ that violate the bound $t(K_n)\leq f_{\binom{n}{2}+1}$ arise for $n=4,5,6$.

\bibliographystyle{alpha}
\bibliography{refs}

\end{document}